\documentclass[reprint, amsmath,amssymb,aps,longbibliography]{revtex4-1}

 \usepackage[bottom]{footmisc}

\usepackage{kpfonts}
\usepackage[T1]{fontenc}
\usepackage[latin9]{inputenc}
\usepackage{mathrsfs}
\usepackage{bm}
\usepackage{amsmath}
\usepackage{amsthm}
\usepackage{amssymb}
\usepackage{stmaryrd}
\usepackage{mathtools}

\edef\ordinarycolon{\mathchar\the\mathcode`: }
\edef\ordinaryequals{\mathchar\the\mathcode`= }

\usepackage{breqn}
\catcode`^=7

\AtBeginDocument{%
  \catcode`^=12
}

\makeatletter

\allowdisplaybreaks

\let\cat@comma@active\@empty

\usepackage[capitalize]{cleveref}

\usepackage{color}
\RequirePackage[dvipsnames,usenames]{xcolor}

\newif\ifnotes
\notestrue

\makeatother

\newcommand{\ba}{\begin{eqnarray}}
\newcommand{\ea}{\end{eqnarray}}

\newcommand{\eq}[1]{\begin{align}#1\end{align}}

\newcommand{\GG}{{{\Sigma}}}

\newcommand{\R}{{\mathbb{R}}}

\newcommand{\dSt}{{\langle\dot{\sigma}(t)\rangle}}

\newcommand{\dSot}{{\langle\dot{\sigma}^\omega(t)\rangle}}
\newcommand{\ooo}{{\omega}}

\newcommand{\dSAt}{{\langle \dot{\sigma}_{K(A;t)}(t) \rangle}}

\newcommand{\NN}{\mathcal{N}}
\newcommand{\ovNN}{{\overline{\NN}}}

\newcommand{\oo}{{\omega}}

\newtheorem{theorem}{Theorem}

\newtheorem{lemma}[theorem]{Lemma}

\begin{document}

\title{Minimal entropy production due to constraints on rate matrix dependencies in multipartite processes}

\author{David H. Wolpert}
\affiliation{Santa Fe Institute, Santa Fe, New Mexico \\
Complexity Science Hub, Vienna\\
Arizona State University, Tempe, Arizona\\
\texttt{http://davidwolpert.weebly.com}}

\begin{abstract}
I consider multipartite processes in which there are constraints on each subsystem's
rate matrix, restricting which other subsystems can directly affect its dynamics. 
I derive a strictly nonzero lower bound on the minimal achievable entropy production rate of the process
in terms of these constraints on the rate matrices of its subsystems.
The bound is based on constructing counterfactual rate matrices, in which some
subsystems are held fixed while the others are
allowed to evolve. This bound is related to the ``learning rate'' of stationary bipartite 
systems, and more generally to the ``information flow'' in bipartite systems.

\end{abstract}

\maketitle

%

\textit{Introduction.---}   
Many systems are naturally modeled as having  two or more interacting subsystems. 
%
Recent research in stochastic
thermodynamics~\cite{seifert2012stochastic,van2015ensemble,esposito2010three,wolpert_thermo_comp_review_2019} 
has started to investigate 
such composite systems~\cite{sagawa2008second,sagawa2009minimal,parrondo2015thermodynamics,horowitz2014thermodynamics,barato_efficiency_2014,ito2013information,hartich_sensory_2016}.
So far, most of the  research has been on the special case of  bipartite processes, i.e.,
systems composed of two co-evolving subsystems, which have zero probability of making a state transition 
simultaneously~\cite{sagawa2008second,sagawa2009minimal,parrondo2015thermodynamics,horowitz2014thermodynamics,hartich_sensory_2016,barato_efficiency_2014,ito2013information,fluct.theorems.partially.masked.shiraishi.sagawa.2015,Bisker_2017,shiraishi_ito_sagawa_thermo_of_time_separation.2015}.  However, 
given that many systems have more than just two interacting subsystems,
research is starting to extend to fully multipartite processes~\cite{horowitz_multipartite_2015,ito2013information,wolpert_book_2018}.

The definition of any composite system specifies which subsystems directly affect the dynamics of which other
subsystems. 
It is now known that just by itself,
such a specification of which subsystem affects which other one
can cause a strictly positive lower bound on the
entropy production rate (EP) of the overall composite system~\cite{wolpert_thermo_comp_review_2019,wolpert2020thermodynamics,Boyd:2018aa}.
This minimal EP 
has sometimes been called ``Landauer
 loss'', because it is the extra EP beyond the minimal amount (namely, zero) implicit
in the Landauer bound~\cite{wolpert_thermo_comp_review_2019,wolpert2020thermodynamics,wolpert_book_review_chap_2019}

Previous analyses of Landauer loss 
focused on scenarios where every subsystem evolves in isolation, without \textit{any} direct coupling to the other subsystems. 
This is a severe limitation of those analyses. As an illustration,
consider a composite system with three subsystems $A, B$ and $C$.
$B$ evolves independently of $A$ and $C$. However, $B$ is continually observed by $C$ 
 as well as $A$. Moreover, suppose that $A$ is really two subsystems, $1$ and $2$. Only subsystem $2$ 
directly observes $B$, whereas subsystem $1$ observes subsystem $2$, e.g., to record a running average of the
values of subsystem $2$ (see~\cref{fig:1}).  

There has been some work on a simplified version of this
scenario, in which subsystem $4$ is absent
and subsystem $3$ is required to be at equilibrium~\cite{hartich_sensory_2016,bo2015thermodynamic}. But this work has focused on
issues other than the minimal EP.

\begin{figure}[tbp]
\includegraphics[width=75mm]{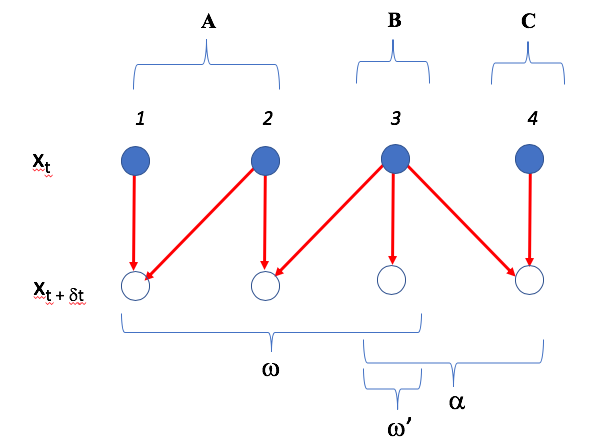}
\caption{Four subsystems, $\{1, 2, 3, 4\}$ interacting in a multipartite process.
The red arrows indicate dependencies in the associated four rate matrices. 
$B$ evolves autonomously, but is continually observed by $A$ and $C$. So the statistical coupling
between $A$ and $C$ could grow with time, even though their rate matrices do not involve one
another. The three overlapping sets indicated at the bottom of the figure specify the three communities of a community structure
for this process. }
\label{fig:1}
\end{figure}

To investigate Landauer loss in these kinds of composite systems, here
I model them as multipartite processes, in which each subsystem evolves according to its own rate
matrix~\cite{horowitz_multipartite_2015}. 
So restrictions on the direct coupling of any subsystem $i$ to the other subsystems are modeled as restrictions
on the rate matrix of subsystem $i$, to only involve a limited
set of other subsystems, called the ``community'' of $i$. (These are instead called ``neighborhoods'' in~\cite{horowitz_multipartite_2015}, 
but that expression already means something in topology, and so I don't use it here.)

In this paper I derive a lower bound
on the Landauer loss rate of composite systems, by deriving an exact equation for that minimal EP rate as a sum of non-negative expressions.
One of those expressions is related to quantities that were earlier considered in the literature. It reduces to what has been called the ``learning rate'' in the special case of stationary bipartite 
systems~\cite{barato_efficiency_2014,Brittain_2017,hartich_sensory_2016}. That expression is also related
to what (in a different context) has been called the ``information flow'' between a pair of subsystems~\cite{horowitz2014thermodynamics,horowitz_multipartite_2015}.

\textit{Rate matrix communities.--- } I write
$\NN$ for a particular set of $N$ subsystems, with finite state spaces
$\{X_i : i = 1, \ldots N\}$. $x$ indicates a vector in $X$, the joint space of $\NN$.
For any $A \subset \NN$, I write $-A := \NN \setminus A$. So for example $x_{-A}$ is the vector of all components of $x$ other than those in $A$.
A distribution over a set of values $x$ at time $t$ is written as $p^X(t)$, with its value for
$x \in X$ written as $p^X_x(t)$, or just $p_x(t)$ for short. 
Similarly, $p^{X|Y}_{x,y}(t)$ is the conditional distribution of $X$ given $Y$ at time $t$, evaluated
for the event $X=x, Y=y$ (which I sometimes shorten to $p_{x|y}(t)$).
I write Shannon entropy as $S(p_X(t))$, $S_t(X)$, or $S^{X}(t)$, as convenient.
I also write the conditional entropy of $X$ given $Y$ at $t$ as $S^{X|Y}(t)$.
I write the Kronecker delta as both $\delta(a, b)$ or $\delta^a_b$.

The joint system evolves as a multi-partite process, there is a set of time-varying stochastic rate matrices, 
$\{K^{x'}_x(i; t) : i = 1, \ldots, N\}$, where for all $i$, $K^{x'}_x(i; t) = 0$ if $x'_{-i} \ne x_{-i}$, and
where the joint dynamics over $X$ is governed by the master equation
\eq{
\frac{d p_x(t)}{dt} &= \sum_{x'} K^{x'}_{x}(t) p_{x'}(t)   \\
	&=  \sum_{x'} \sum_{i \in \NN} K^{x'}_{x}(i; t) p_{x'}(t)
}
Note that each subsystem can be driven by its own external work reservoir, according to a time-varying protocol.
For any $A \subseteq \NN$ I define
\eq{
K^{x'}_{x}(A; t) &:= \sum_{i \in A} K^{x'}_{x}(i; t)
\label{eq:1}
}

Each subsystem $i$'s marginal distribution evolves as
\eq{
\label{eq:5a}
\frac{d p_{x_i}(t)}{dt} &= \sum_{x_{-i}}  \sum_{x'} \sum_j K^{x'}_{x}(j; t) p_{x'}(t) \\
	&= \sum_{x'} K^{x'_i, x'_{{-i}}}_{x_i, x'_{-i}}(i; t) p_{x'}(t)
\label{eq:6}
}
due to the multipartite nature of the process 
\footnote{To see this, note that if $x'_i \ne x_i$, then the only way for $ K^{x'}_{x}(j; t) p_{x'}(t)$ to be nonzero
is if $x'_{-i} = x_{-i}$ and $j = i$. If instead $x'_i = x_i$, $j$ can differ from $i$.
However, if $j \ne i$ then the sum over $x_{-i}$ in \cref{eq:5a} runs over all values of $x_j$. By normalization
of the rate matrix $K^{x'}_{x}(j; t)$, that sum must equal zero.}.
\cref{eq:6} shows that in general the marginal distribution $p_{x_i}$ will not evolve according to a 
continuous-time Markov chain (CTMC) over $\Delta_{X_i}$.

For each subsystem $i$, I write $r(i; t)$ for any set of subsystems at time $t$ that includes $i$
where we can write
\eq{
K^{x'}_x(i; t) = K^{x'_{r(i;t)}}_{x_{r(i;t)}}(i; t) \delta(x'_{-r(i;t)}, x_{-r(i;t)})
\label{eq:def_community_rate}
} 
for an appropriate set of functions $K^{x'_{r(i;t)}}_{x_{r(i;t)}}(i; t)$.  In general, $r(i;t)$ is not uniquely defined,
since I make no requirement that it be minimal.
I refer to the elements of $r(i;t)$ as the \textbf{leaders} of $i$ at time $t$. Note that the leader relation need
not be symmetric. A \textbf{community} $\oo$ at time $t$ is a set of subsystems such that $i \in \oo$ implies that $r(i;t) \subseteq \oo$.
Any intersection of two communities is a community, as is any union of two communities. Intuitively, a community
is any set of subsystems whose evolution is independent of the states of the subsystems outside
the community (although in general, the evolution of those external subsystems may depend on the states of subsystems in the community).

A specific set of communities that covers $\NN$ and is closed under intersections is a \textbf{community structure}.
A \textbf{community topology} is  a community structure that is closed under unions, with the communities of the structure being the open
sets of the topology. However, in general, unless explicitly stated otherwise, any community structure being discussed does not
have $\NN$ itself as a member.

As an example of these definitions,~\cite{hartich_stochastic_2014,barato_efficiency_2014,hartich_sensory_2016}
investigate a special type of bipartite system, where the ``internal''  subsystem $B$ observes the ``external'' subsystem $A$,
but cannot affect the dynamics of that external subsystem. So
$A$ is its own community, evolving independently of $B$, while $B$ is not its own community; its
dynamics depends on the state of $A$ as well as its own state. 
Another example of these definitions is illustrated in \cref{fig:1}.

For simplicity, from now on I assume that the set of communities doesn't change with $t$.
Accordingly I shorten $r(i;t)$ to $r(i)$. For any community $\oo$ I write
\eq{
K^{x'_\oo}_{x_\oo}(\oo; t) &:= \sum_{i \in \oo} K^{x'_{\oo}}_{x_{\oo}}(i; t)
}
So $K^{x'}_x(\oo; t) = K^{x'_\oo}_{x_\oo}(\oo; t)  \delta(x'_{-\oo}, x_{-\oo})$, by \cref{eq:1,eq:def_community_rate}.

At any time $t$, for any community $\omega$, $p_{x_{\omega}}(t)$ evolves as a CTMC with rate
matrix $K^{x'_{\omega}}_{x_\omega}(\omega; t)$:
\eq{
\frac{d p_{x_\oo}(t)}{dt} 
	&= \sum_{x'_\oo} K^{x'_\oo}_{x_\oo}(\oo; t) p_{x'_\oo}(t)
\label{eq:15aa}
}
(See SI.)
So a community evolves according
to a self-contained CTMC, in contrast to the general case of a single subsystem (cf. \cref{eq:6}).

I assume that each subsystem is attached to at most one thermal reservoir, and that
all such reservoirs have the same temperature~\cite{horowitz_multipartite_2015}.
Accordingly, the expected entropy flow (EF) rate of any community $\oo \subseteq \NN$ at time $t$ is
\eq{
\langle\dot{Q}^{\oo;K} (t)\rangle &= \sum_{x'_\oo,x_\oo} K^{x'_\oo}_{x_\oo}(\oo; t) p_{x'_\oo}(t) \ln \left[\dfrac{K^{x'_\oo}_{x_\oo}(\oo; t)}
				{K^{x_\oo}_{x'_\oo}(\oo; t)}\right]  
\label{eq:4}
}
which I often shorten to $\langle\dot{Q}^{\oo} (t)\rangle$~\cite{van2015ensemble,esposito2010three}.
(Note that this is entropy flow \textit{from} $\oo$ into the environment.) 
Make the associated definition that the expected EP rate of $\oo$ at time $t$ is
\eq{
\label{eq:15}
\langle\dot{\sigma}^{\oo; K} (t) \rangle &= \dfrac{d S^\oo(t)}{dt} + \langle\dot{Q}^\oo (t)\rangle    \\
 &\!\!\!\!\!\!\!=\sum_{x'_\oo,x_\oo} K^{x'_\oo}_{x_\oo}(\oo; t) p_{x'_\oo}(t) \ln \left[\dfrac{K^{x'_\oo}_{x_\oo}(\oo; t) p_{x'_\oo}(t)}
				{K^{x_\oo}_{x'_\oo}(\oo; t) p_{x_\oo}(t)}\right]
\label{eq:5}
}
which I often shorten to $\langle\dot{\sigma}^{\oo} (t)\rangle$.

I refer to $\langle\dot{\sigma}^{\oo} (t)\rangle$ as a \textbf{local} EP rate, and define the \textbf{global} EP rate as
$\dSt := \langle\dot{\sigma}^{\NN} (t) \rangle$.
For any community $\oo$, $\dSot \ge 0$, since
$\dSot$ has the usual
form of an EP rate of a single system. In addition, that lower bound of $0$ is achievable, e.g., if
$K^{x'_\oo}_{x_\oo}(\oo; t) p_{x'_\oo}(t) = K^{x_\oo}_{x'_\oo}(\oo; t) p_{x_\oo}(t)$ at time $t$ for all $x_\oo, x'_\oo$.

It is worth comparing the local EP rate to similar quantities that have been investigated in the literature.
In contrast to $\dSot$, the quantity ``$\sigma_X$'' introduced in the analysis of  (autonomous) bipartite systems
in~\cite{fluct.theorems.partially.masked.shiraishi.sagawa.2015} is the EP of a single trajectory, integrated over 
time. More importantly, its expectation can be
negative, unlike (the time-integration of) $\dSot$. 
On the other hand, the quantity ``$\dot{S}^X_i$'' considered 
in the analysis of bipartite systems in~\cite{horowitz2014thermodynamics} is a proper expected EP rate, and so
is non-negative. However, it (and its extension considered  in~\cite{horowitz_multipartite_2015}) is one term in a decomposition 
of the expected EP rate generated by a {single} community. It does not concern the EP rate of an entire community in a system with multiple communities.
Finally,  the quantity ``$\sigma_\Omega$'' considered 
in~\cite{fluct.theorems.partially.masked.shiraishi.sagawa.2015} is also
non-negative. However, it gives the total EP rate generated by a subset of all possible global state transitions, rather than
the EP rate of a community \footnote{It is
possible to choose that subset of state transitions so that $\sigma_\Omega$ concerns all transitions
in which one particular subsystem changes state while all others do not. In this case, like the quantity
 $\dot{S}^X_i$ considered in~\cite{horowitz2014thermodynamics}, $\sigma_\Omega$ is a single term in the decomposition
of the EP rate generated by a single community.}.

\textit{EP bounds from counterfactual rate matrices.---} To analyze the minimal EP rate in multipartite processes, we 
need to introduce two more definitions. First,
given any function $f : \Delta_X \rightarrow \R$ and any $A \subset \NN$ (not necessarily a community), define
the \textbf{$A$-(windowed) derivative} of $f(p(t))$  under rate matrix $K(t)$ as 
\eq{
\dfrac{d^{A;K(t)} f(p(t)}{dt}&=\sum_{x,x'} K^{x'}_x(A; t) p_{x'}(t) \dfrac{ \partial f(p_x(t))}{\partial p_x(t)}
}
(See \cref{eq:1}.) Intuitively, this is what the derivative of $f(p(t))$ would be
if (counterfactually) only the subsystems in $A$ were allowed to change their states.

In particular, 
the $A$-derivative of the conditional entropy of $X$ given $X_A$ is 
\eq{
\dfrac{d^{A;K(t)}}{dt} S^{X | X_{A}}(p(t)) &= -\sum_{x,x'} K^{x'}_x(A; t) p_{x'}(t) \ln p_{x | x_A}(t)
\label{eq:27}
}
which I sometimes write as just $\dfrac{d^{A}}{dt} S^{X | X_{A}}(t)$.  (See Eq.\,4 in~\cite{horowitz_multipartite_2015} for a similar quantity.)
$\dfrac{d^{A}}{dt} S^{X | X_{A}}(t)$ measures how quickly the statistical coupling between $X_A$ and $X_{-A}$ changes
with time, if rather than evolving under the actual rate matrix, the system evolved under a counterfactual rate matrix, in which $x_{-A}$ is not allowed to change. 
In the SI it is shown that in the special case that
$A$ is a community, $\dfrac{d^{A}}{dt} S^{X | X_{A}}(t)$
is the derivative of the negative mutual information between $X_A$ and $X_{-A}$, {under the counterfactual rate
matrix} $K(A;t)$, and is therefore non-negative \footnote{In~\cite{horowitz2014thermodynamics,horowitz_multipartite_2015}, the 
$A$-derivative of the mutual information between $A$ and $\NN\setminus A$,
$d^A I(X_A; X_{\NN \setminus A}) /dt$, is interpreted as the
``information flow'' from $\NN\setminus A$ to $A$.  (See Eq.\,12 in~\cite{horowitz_multipartite_2015}.) 
However, in the scenarios considered in this paper $A$ will always
be a community. Therefore none of the subsystems in $A$ will evolve in a way
directly dependent on the state of any subsystem in $\NN\setminus A$ (nor vice-versa). So the fact that $d^A I(X_A; X_{\NN \setminus A}) /dt < 0$
will not indicate that information in $A$ concerning $\NN \setminus A$
``flows'' between $A$ and $\NN\setminus A$ in any sense.
In the current context, it would be more accurate to refer to $-d^A S(X | X_A) /dt$ as
the ``forgetting rate'' of $A$ concerning $\NN \setminus A$, than as ``information flow''. $d^A I(X_A; X_{\NN \setminus A}) /dt$ is also related to what is
termed ``nostalgia'' in~\cite{still2012thermodynamics}. However, that paper considers discrete-time rather than continuous-time processes,
where subsystems are required to start in thermal equilibrium.
}. 


The second definition we need is a variant of $\langle\dot{\sigma}^{\oo; K} (t) \rangle$, which
will be indicated by using subscripts rather than superscripts. 
For any $A \subseteq B \subseteq \NN$ where $B$ is a community (but $A$ need not be),
\eq{
\langle {\dot{\sigma}}_{K(A;t); B} \rangle
	&:= \sum_{x,x'\in X_B} K^{x'}_{x}(A; t) p_{x'}(t) \ln \bigg[\dfrac{K^{x'}_{x}(A; t) p_{x'}(t)}
				{K^{x}_{x'}(A; t) p_{x}(t)}\bigg]
}
which I abbreviate as $\dSAt$ when $B = \NN$. 
$\dSAt$ is a global EP rate, only evaluated under the counterfactual rate matrix $K(A; t)$. Therefore it is non-negative. 
In contrast, $\langle\dot{\sigma}^{\oo; K} (t) \rangle$ is a local EP rate. In the special case that $A = \oo$ is a community, these two EP rates 
are related by
$\langle {\dot{\sigma}}_{K(A;t)}(t) \rangle = \langle\dot{\sigma}^{A; K} (t) \rangle + \frac{d^A}{dt} S^{X|X_A}(t)$
(see \cref{eq:a2} in the SI).

In the SI it is shown that
for any pair of communities, $\oo$ and $\oo' \subset \oo$, 
\eq{
\!\!\!\!\! \langle {\dot{\sigma}}^{\oo;K(\oo;t)}(t)\rangle &= \langle {\dot{\sigma}}^{\oo';K(\oo;t)}(t)\rangle + 
	\langle {\dot{\sigma}}_{K(\oo \setminus \oo';t);\oo}(t)\rangle \nonumber \\
\label{eq:30a}
			&\qquad\qquad\qquad +  \dfrac{d^{\oo'}}{dt}S^{X_\oo | X_{\oo'}}(t)
}
(See \cref{fig:1} for an illustration of such a pair of communities $\oo, \oo' \subset \oo$.)
The first term on the RHS
is the EP rate arising from the subsystems within community $\oo'$, and the second term is
the ``left over'' EP rate from the subsystems that are in $\oo$ but not in
$\oo'$. The third term is a time-derivative of the conditional entropy between
those two sets of subsystems. All three of these terms are non-negative, so each of them
provides a lower bound on the EP rate. 

\cref{eq:30a} is the major result of this paper. 
In particular, setting $\oo = \NN$ and then consolidating notation by rewriting $\oo'$ as $\oo$, \cref{eq:30a} shows that for any community $\oo$,
\eq{
\label{eq:30b}
\dSt  &= \langle {\dot{\sigma}}^{\oo}(t)\rangle + 
	\langle {\dot{\sigma}}_{K(\NN \setminus \oo;t)}(t)\rangle  + \dfrac{d^{\oo;K(\NN;t)}}{dt} S^{X | X_{\oo}}(t) 	\\
	& \ge \dfrac{d^{\oo}}{dt}S^{X | X_{\oo}}(t) 
\label{eq:30}
}
(where the shorthand notation has been used).

As an example of \cref{eq:30}, consider again the type of bipartite process analyzed
in~\cite{hartich_stochastic_2014,barato_efficiency_2014,hartich_sensory_2016}. Suppose we set $\oo$ to contain only what in~\cite{hartich_stochastic_2014} is called the ``external'' 
subsystem. Then if we also make the assumption of those papers that
the full system is in a stationary state, $d S^X / dt = d S^{X_\oo} /dt = d S^{X_{-\oo}} / dt = 0$. So by \cref{eq:27},
\eq{
\dfrac{d^{\oo}}{dt}S^{X | X_{\oo}}(t) &= -\dfrac{d^{-\oo}}{dt}S^{X | X_{-\oo}}(t) 
} 
(The RHS is called the ``learning rate'' of the internal subsystem about the external subsystem --- see Eq.\,(8) in~\cite{Brittain_2017},
and note that the rate matrix is normalized.)

So in this scenario, \cref{eq:30} above reduces to Eq.\,7 of~\cite{barato_efficiency_2014}, which lower-bounds the global EP rate
by the learning rate. However, \cref{eq:30} lower-bounds the global EP rate 
even if the system is not in a stationary state, which need be the case with the learning rate
\footnote{Recall from the discussion 
above of the scenario considered in~\cite{barato_efficiency_2014} that while the external subsystem $\oo$ is its own community, 
the internal subsystem is {not}. This means that in general, if the full system is not in a stationary state, then the learning rate 
of the internal subsystem about the external subsystem (as defined in~\cite{barato_efficiency_2014}) cannot be expressed as
$d^\omega S^{X | X_{\omega}}(t) /dt$ with $\oo$ being a community.}. More generally, \cref{eq:30b} applies to arbitrary multipartite
processes, not just those with two subsystems, and is an exact equality rather than just a bound. 

In some situations we can get an even more refined decomposition of EP rate by substituting \cref{eq:30a} into \cref{eq:30b} 
to expand the first EP rate on the RHS of \cref{eq:30b}. This gives a larger lower bound on $\dSt$
than the one in \cref{eq:30}.
For example, if $\oo$ and $\oo' \subset \oo$ are both communities under $K(t)$, then 
\eq{
\label{eq:22}
\dSt &= \langle {\dot{\sigma}}^{\oo'}(t)\rangle + 
	\langle {\dot{\sigma}}_{K(\oo \setminus \oo';t);\oo}(t)\rangle  + \dfrac{d^{\oo'}}{dt} S^{X_\oo | X_{\oo'}}(t)  \\
	&\qquad	+ \langle {\dot{\sigma}}_{K(\NN \setminus \oo;t)}(t)\rangle + \dfrac{d^{\oo}}{dt} S^{X | X_{\oo}}(t)  \nonumber \\
	&\ge    \dfrac{d^{\oo}}{dt} S^{X | X_{\oo}}(t) + \dfrac{d^{\oo'}}{dt} S^{X_\oo | X_{\oo'}}(t)
\label{eq:37}
}
Both of the terms on the RHS in \cref{eq:37} are non-negative. In addition, both can
be evaluated without knowing the detailed physics occurring {within} communities $\oo$
or $\oo'$, only knowing how the statistical coupling {between} communities evolves with time.

This can be illustrated with the scenario depicted in \cref{fig:1}. Using the
communities $\oo$ and $\oo'$ specified there, \cref{eq:37} says that the global EP rate is lower-bounded
by the sum of two terms. The
first is the derivative of the negative mutual information between subsystem $4$ and the first three
subsystems, if subsystem $4$ were held fixed. The second is the derivative of the negative mutual information 
between subsystem $3$ and the first two subsystems, if those two subsystems were held fixed.

Alternatively, suppose that $\oo$ is a community under $K$, 
and that some set of subsystems $\alpha$ is a community under 
$K(\NN \setminus \oo; t)$.
Then since the term $\langle {\dot{\sigma}}_{K(\NN \setminus \oo;t)}(t)\rangle$ in \cref{eq:30b} is 
a global EP rate over $\NN$ under rate matrix $K(\NN \setminus \oo;t)$, we can again feed \cref{eq:30a} into \cref{eq:30b},
(this time to expand the second rather than first term on the RHS of \cref{eq:30b})  to get
\eq{
\dSt &= \langle {\dot{\sigma}}^{\oo;K(t)}(t)\rangle + \dfrac{d^{\oo}}{dt} S^{X | X_{\oo}}(t) 
		+ \langle {\dot{\sigma}}^{\alpha; {K(\NN \setminus \oo; t)}}(t)\rangle   \nonumber \\
\label{eq:37b}
	&\;\;\;\;+ \langle {\dot{\sigma}}_{K((\NN\setminus \oo)\setminus \alpha ;t);\NN\setminus \oo}(t)\rangle
		+ \dfrac{d^{\alpha;K(\NN\setminus \oo; t)}}{dt} 
				S^{X_{\NN} | X_{\alpha}}(t)  \\
	&\ge \;\;\dfrac{d^{\oo}}{dt}S^{X | X_{\oo}}(t)  + \dfrac{d^{\alpha;K(\NN\setminus \oo; t)}}{dt} 
				S^{X_{\NN} | X_{ \alpha}}(t)
\label{eq:37a}
}
The RHS of \cref{eq:37a} also exceeds the bound in \cref{eq:30}, by the negative $\alpha$-derivative of the
mutual information between $X_{\NN\setminus\alpha}$ and $X_{\alpha}$,  under the rate matrix $K(\NN \setminus \oo;t)$.

\textit{Example.---} Depending on the full community structure, we may be able to combine \cref{eq:30a,eq:37b} into an even larger 
lower bound on the global EP rate than \cref{eq:37a}. To illustrate this, return to the
scenario depicted in \cref{fig:1}. 
Take $\oo = \{1, 2, 3\}$ and $\alpha = \{3, 4\}$, as indicated in that figure. Note that the four sets $\{1\}, \{2\}, \{3\}, \{3, 4\}$
form a  community structure of
\eq{
K^{x'}_{x}(\NN \setminus \oo;t) &= K^{x'}_{x}(\{4\}; t) 
\label{eq:ex:1}
} 
since under $K^{x'}_x(\{4\}; t)$, neither subsystem $1, 2$ nor $3$ changes its state.
So $\alpha$ is a member of a community structure of $K(\NN \setminus \oo;t)$, and we can apply  \cref{eq:37b}.

The first term in \cref{eq:37b}, $\langle {\dot{\sigma}}^{\oo}(t)\rangle$, 
is the local EP rate that would be jointly generated by the set of
three subsystems $\{1, 2, 3)$, if they evolved in isolation from the other subsystem, under the self-contained rate matrix
\eq{
K^{x'}_x(\{1,2,3\};t) &= K^{x'_1, x'_2, x'_3}_{x_1, x_2, x_3}(1, t) + K^{x'_1, x'_2, x'_3}_{x_1, x_2, x_3}(2, t) + K^{x'_1, x'_2, x'_3}_{x_1, x_2, x_3}(3, t)
}

The third term in \cref{eq:37b} is the local  EP rate that would be jointly generated by the 
two subsystems $\{3, 4\}$, if they evolved in isolation from the other two subsystems, but rather than do so
under the rate matrix $K(\alpha; t) = K(\{3, 4\}; t)$,  they did so under the rate matrix $K^{x'}_{x}(\NN \setminus \oo;t)$ 
given in \cref{eq:ex:1}.
(Note that $K^{x'}_{x}(\NN \setminus \oo;t) = 0$ if $x'_3 \ne x_3$, unlike $K^{x'}_x(\{3, 4\}; t)$.)
The fourth term in \cref{eq:37b} is the global EP rate that would be generated by evolving all four subsystems under the rate matrix 
for the subsystems in $(\NN \setminus \oo) \setminus \alpha$. But there are no subsystems in that set. So this fourth term is zero.


Both that first and third term in \cref{eq:37b} are non-negative. The remaining two terms
 -- the second and the fifth in \cref{eq:37b} --- 
are also non-negative. However,  in 
contrast to the terms just discussed, these two depend only
on derivatives of mutual informations. Specifically, the second term in \cref{eq:37b} is the negative 
derivative of the mutual information between the joint random variable $X_{1, 2, 3}$ and $X_4$, under the rate
matrix $K^{x'}_x(\{1,2,3\};t)$.
Next, since $ \NN \setminus \alpha = \{1, 2\}$, the fifth term is the negative 
derivative of the mutual information between $X_{1,2}$ and $X_{3,4}$, under the rate
matrix given by windowing $\alpha$ onto $K(\NN\setminus \oo;t)$, i.e., under the rate matrix $K^{x'}_{x}(\{4\}; t)$.

Recalling that $\oo := \{1, 2, 3\}, \alpha:=\{3,4\}$ and defining $\gamma :=\{4\}$, we can 
combine these results to express the global EP rate of the system illustrated in \cref{fig:1} in terms of the rate matrices
of the four subsystems:
\eq{
&\dSt = \sum_{x'_\oo,x_\oo} K^{x'_\oo}_{x_\oo}(\oo; t) p_{x'_\oo}(t) \ln \left[\dfrac{K^{x'_\oo}_{x_\oo}(\oo; t) p_{x'_\oo}(t)}
				{K^{x_\oo}_{x'_\oo}(\oo; t) p_{x_\oo}(t)}\right] \nonumber \\
	& \qquad \qquad+ \sum_{x'_\alpha,x_\alpha} K^{x'_\alpha}_{x_\alpha}(\gamma; t) p_{x'_\alpha}(t) \ln \left[\dfrac{K^{x'_\alpha}_{x_\alpha}(\gamma; t) p_{x'_\alpha}(t)}
				{K^{x_\alpha}_{x'_\alpha}(\gamma; t) p_{x_\alpha}(t)}\right]   \nonumber \\
	& \; -\sum_{x,x'} \left[K^{x'}_x(\oo; t) p_{x'}(t) \ln p_{x | x_\oo}(t) + K^{x'}_x(\gamma; t) p_{x'}(t) \ln p_{x | x_\alpha}(t)\right]
\label{eq:25}
}
All five terms on the RHS of \cref{eq:25} are non-negative.
Translated to this scenario, previous results concerning learning rates consider the special case of a stationary state
$p_x(t)$, and only tell us that the global EP rate is bounded by the fourth term on the RHS of \cref{eq:25}:
\eq{
&\dSt \ge -\sum_{x,x'} K^{x'}_x(\oo; t) p_{x'}(t) \ln p_{x | x_\oo}(t)
}

Finally note that we also have a community $\oo' = \{3\}$ which is a
proper subset of both $\oo$ and $\alpha$. So, for example, we can plug this $\oo'$ into \cref{eq:30a} 
to expand the first term in \cref{eq:37b}, $\langle {\dot{\sigma}}^{\oo;K(\oo;t)}(t)\rangle$, replacing it with 
the sum of three terms. The first of these three new terms, $\langle {\dot{\sigma}}^{\oo';K(\oo;t)}(t)\rangle$, is
the local EP rate generated by subsystem $\{3\}$ evolving in isolation from all the other subsystems.
The second of these new terms,  $\langle {\dot{\sigma}}_{K(\oo \setminus \oo';t);\oo}(t)\rangle$,
is the EP rate that would be generated if the set of three subsystems $\{1, 2, 3\}$ evolved in isolation from
the remaining subsystem, $4$, but under the rate matrix 
\eq{
K(\oo \setminus \oo'; t) &= K^{x'_1, x'_2, x'_3}_{x_1, x_2, x_3}(1; t) + K^{x'_1, x'_2, x'_3}_{x_1, x_2, x_3}(2; t)
}
The third new
term is the negative derivative of the mutual information between 
$X_{1, 2}$ and $X_3$, under rate matrix $K(\oo; t)$. All three of these new terms are non-negative.

\textit{Discussion.---} 
There are other decompositions of the global EP rate which are of interest, but don't always
provide non-negative lower-bounds on the EP rate. One of them based on the inclusion-exclusion principle is discussed in \cref{app:in_ex_sums}.
Future work involves combining these (and other) decompositions, to get even larger lower bounds.

$ $

 I would like to thank Sosuke Ito, Artemy Kolchinsky, Kangqiao Liu, Alec Boyd, Paul Riechers, 
and especially Takahiro
Sagawa for stimulating discussion. This work was supported by the Santa Fe Institute, 
Grant No. CHE-1648973 from the US National Science Foundation and Grant No. FQXi-RFP-IPW-1912 from the FQXi foundation.
The opinions expressed in this paper are those of the author and do not necessarily 
reflect the view of the National Science Foundation.

\bibliographystyle{amsplain}

\newcommand{\arXiv}[2]{\href{http://arxiv.org/abs/#1}{arXiv:#1 #2}}
\providecommand{\bysame}{\leavevmode\hbox to3em{\hrulefill}\thinspace}
\providecommand{\MR}{\relax\ifhmode\unskip\space\fi MR }
\providecommand{\MRhref}[2]{%
  \href{http://www.ams.org/mathscinet-getitem?mr=#1}{#2}
}
\providecommand{\href}[2]{#2}

\appendix

\section{Proof of \cref{eq:15aa}}
\label{app:proof_eq_15aa}

Write
\eq{
\frac{d p_{x_\oo}(t)}{dt}  	&= \sum_{x_{-\oo}} \sum_{x'} \sum_j K^{x'}_{x}(j;t)p_{x'}(t)  \nonumber \\
	&= \sum_{x'} p_{x'}(t) \left[\sum_{j \in \oo}  \sum_{x_{-\oo}} K^{x'}_{x}(j;t) + \sum_{j \not\in \oo}  \sum_{x_{-\oo}} K^{x'}_{x}(j;t)\right] 
}
If $j \not \in \oo$, then a sum over all $x_{-\oo}$ in particular runs over all $x_j$. Therefore we get
\eq{
\frac{d p_{x_\oo}(t)}{dt}  	&= \sum_{x'} p_{x'}(t)  \sum_{x_{-\oo}} \sum_{j \in \oo} K^{x'}_{x}(j;t)   
}

Using the fact that we have a multipartite process and then the
fact that $\oo$ is a community, we can expand this remaining expression as
\eq{
\sum_{x'}  \sum_{j \in \oo} K^{x'_\oo, x'_{-\oo}}_{x_\oo, x'_{-\oo}}(j;t) p_{x'_\oo, x'_{-\oo}}(t)
	&= \sum_{x'}  \sum_{j \in \oo} K^{x'_\oo}_{x_\oo}(j;t) p_{x'_\oo, x'_{-\oo}}(t) \nonumber \\
	&= \sum_{x'_\oo}  \sum_{j \in \oo} K^{x'_\oo}_{x_\oo}(j;t) p_{x'_\oo}(t) 
}
To complete the proof plug in the definition of $K^{x'_\oo}_{x_\oo}(\oo; t)$.

\section{Expansions of EP rates in multipartite processes}
\label{app:proof_eq_11}


\begin{widetext}

\begin{lemma}
Suppose we have a multipartite process over a set of systems $\NN$ defined by a set of rate matrices $\{K^{x'}_x(i; t)\}$
and a subset $A \in \NN$. Then
\eq{
\label{eq:a0}
\sum_ {x,x'} K^{x'}_{x}(A;t)  p_{x'}(t) \ln \left[\dfrac{K^{x'}_{x}(t)}
				{K^{x}_{x'}(t)}\right] &= \sum_{x,x'} K^{x'}_x(A; t) p_{x'}(t) \ln \left[\dfrac{K^{x'}_x(A; t)}
				{K^{x}_{x'}(A; t) }\right]	\\
	&= \sum_{i \in A, x,x'} K^{x'}_x(i; t) p_{x'}(t) \ln \left[\dfrac{K^{x'}_x(i; t)}
				{K^{x}_{x'}(i; t) }\right]
\label{eq:a1}
} 
If in addition $A$ is a community under $K$, then
we can also write the quantity in \cref{eq:a0} as
\eq{
\sum_{x_A,x'_A} K^{x'_A}_{x_A}(A;t)  p_{x'_A}(t) \ln \left[\dfrac{K^{x'_A}_{x_A}(A;t)}
				{K^{x_A}_{x'_A}(A; t)}\right]
\label{eq:a2}
}
\label{lemma:a1}
\end{lemma}

\begin{proof}
Invoking the multipartite nature of the process allows us to write 
\eq{
& \sum_{x,x'} K^{x'_A}_{x_A}(A;t) \delta^{x'_{-A}}_{x_{-A}} p_{x'}(t) \ln \left[\dfrac{K^{x'}_{x}(t)}
				{K^{x}_{x'}(t) }\right] \;= \;
	\sum_{i \in A, x',x} K^{x'}_x(i; t) p_{x'}(t) \ln \left[\dfrac{K^{x'}_x(t)}
				{K^{x}_{x'}(t) }\right]  \nonumber \\
&\qquad\qquad = \sum_{i \in A, x_i,x'_i\ne x_i,x_{-i}} K^{x'_i,x_{-i}}_{x_i,x_{-i}}(i; t) p_{x'_i,x_{-i}}(t) \ln \left[\dfrac{\sum_j K^{x'_i,x_{-i}}_{x_i,x_{-i}}(j;t)}
				{\sum_j K^{x_i,x_{-i}}_{x'_i,x_{-i}}(j;t) }\right]  +  
		\sum_{i \in A, x_i,x_{-i}} K^{x_i,x_{-i}}_{x_i,x_{-i}}(i; t) p_{x_i,x_{-i}}(t) \ln \left[\dfrac{\sum_j K^{x_i,x_{-i}}_{x_i,x_{-i}}(j;t)}
				{\sum_j K^{x_i,x_{-i}}_{x_i,x_{-i}}(j;t) }\right]   \nonumber \\
\label{eq:b3}
&\qquad\qquad = \sum_{i \in A, x_i,x'_i\ne x_i,x_{-i}} K^{x'_i,x_{-i}}_{x_i,x_{-i}}(i; t) p_{x'_i,x_{-i}}(t) \ln \left[\dfrac{\sum_{j \in A} K^{x'_i,x_{-i}}_{x_i,x_{-i}}(j;t)}
				{\sum_{j \in A} K^{x_i,x_{-i}}_{x'_i,x_{-i}}(j;t) }\right]  
\\
	&\qquad\qquad=  \sum_{i \in A, x_i,x'_i,x_{-i}} K^{x'_i,x_{-i}}_{x_i,x_{-i}}(i; t) p_{x'_i,x_{-i}}(t) \ln \left[\dfrac{K^{x'_i,x_{-i}}_{x_i,x_{-i}}(i;t)}
				{K^{x_i,x_{-i}}_{x'_i,x_{-i}}(i;t) }\right]
\label{eq:b4}
}
\cref{eq:b3} establishes \cref{eq:a0} and \cref{eq:b4} establishes \cref{eq:a1}. 

To establish \cref{eq:a2}, use the hypothesis that $A$ is a community to expand
\eq{
\sum_{x,x'} K^{x'}_x(A; t) p_{x'}(t) \ln \left[\dfrac{K^{x'}_x(A; t)} {K^{x}_{x'}(A; t) }\right] &= 
	\sum_{x,x'} K^{x'}_x(A; t) \delta^{x'_{-A}}_{x_{-A}} p_{x'}(t) \ln \left[\dfrac{K^{x'_A}_{x_A}(A; t)\delta^{x'_{-A}}_{x_{-A}}}
				{K^{x}_{x'}(A; t) \delta^{x_{-A}}_{x'_{-A}}}\right] \\
	&= 	\sum_{x,x'} K^{x'_A}_{x_A}(A; t) \delta^{x'_{-A}}_{x_{-A}} p_{x'}(t) \ln \left[\dfrac{K^{x'_A}_{x_A}(A; t)}
				{K^{x}_{x'}(A; t) }\right] \\
	&= 	\sum_{x_A,x_{-A},x'_A} K^{x'_A}_{x_A}(A; t) p_{x'_A,x_{-A}}(t) \ln \left[\dfrac{K^{x'_A}_{x_A}(A; t)}
				{K^{x}_{x'}(A; t) }\right] \\
	&= 	\sum_{x_A,x'_A} K^{x'_A}_{x_A}(A; t) p_{x'_A}(t) \ln \left[\dfrac{K^{x'_A}_{x_A}(A; t)}
				{K^{x}_{x'}(A; t) }\right]
}

\end{proof}

\section{Proof that  if $A$ is a community, then $\dfrac{d^{A}}{dt} S^{X | X_{A}}(t) \ge 0$}
\label{app:mask_non_negative}

If $A$ is a community, then 
\eq{
 -\sum_{x,x'} K^{x'}_x(A; t) p_{x'}(t) \ln p_{x_A}(t) &= -\sum_{x_A,x'_A} K^{x'_A}_{x_A}(A; t) p_{x'}(t) \ln p_{x_A}(t) \nonumber \\
	 &= \dfrac{d}{dt} S^{X_A}(t)
}
We can combine this with \cref{eq:27} to expand
\eq{
\dfrac{d^{A}}{dt} S^{X | X_{A}}(t) &= \dfrac{d}{dt} S^{X}(t) -  \dfrac{d}{dt} S^{X_A}(t)
}
(Note that this expansion need not hold if $A$ is not a community.)

Suppose we could also establish that because subsystems outside of $A$ don't evolve under $K(A; t)$, then $S^{X_{-A}}(t)$ doesn't change in time, i.e., that
\eq{
\dfrac{d^A}{dt} S^{X_{-A}}(t) &= 0
\label{eq:d0}
}
This would then imply that 
\eq{
\dfrac{d^{A}}{dt} S^{X | X_{A}}(t) &= \dfrac{d^{A}}{dt} I^{X_{-A} | X_{A}}(t)
}
the windowed time-derivative of the mutual information between the communities in $A$ and those outside of it.
However, $S^{X_{-A}}(t)$ is given by marginalizing $p^X(t)$ down to the subsystems in $-A$, by averaging over $x_A$.
In general, if $A$ is not a community, those subsystems are statistically coupled with the ones in $A$. So as the subsystems in $A$
evolve,  $S^{X_{-A}}(t)$  might change, i.e., \cref{eq:d0} may not hold.

This turns out not to be a problem when $A$ is a community. To see this, first simplify notation by 
using $P$ rather than $p$ to indicate joint distributions that would evolve if $K(t)$ were
replaced by the counterfactual rate matrix $K(A;t)$, starting from $p_x(t)$. By definition,
\eq{
K^{x_A(t), x_{-A}(t)}_{x_A(t+\delta t), x_{-A}(t+\delta t)}(A;t) &=
\lim_{\delta t \rightarrow 0} \dfrac{\delta^{x_A(t), x_{-A}(t)}_{x_A(t+\delta t), x_{-A}(t+\delta t)}
		- P\left(x_A(t+\delta t), x_{-A}(t+\delta t) \;|\; x_A(t), x_{-A}(t)\right)} {\delta t}
\label{eq:d1}
}
However, since by hypothesis $A$ is a community, 
\eq{
K^{x_A(t), x_{-A}(t)}_{x_A(t+\delta t), x_{-A}(t+\delta t)}(A;t) &= K^{x_A(t)}_{x_A(t+\delta t)}(A;t)\delta^{x_{-A}(t)}_{x_{-A}(t+\delta t)}
}
Plugging this into \cref{eq:d1} and summing both sides over $x_A(t+\delta t)$ shows that to leading order in $\delta t$,
\eq{
P\left(x_{-A}(t+\delta t) \;|\; x_A(t), x_{-A}(t)\right) &= \delta^{x_{-A}(t+ \delta t)}_{x_{-A}(t)}
\label{eq:d2}
}
\cref{eq:d2} in turn implies that to leading order in $\delta t$,
\eq{
P(x(t+\delta t) \;|\; x(t)) &= P\left(x_A(t+\delta t) \;|\; x_{-A}(t+\delta t), x_A(t), x_{-A}(t)\right)
			P\left(x_{-A}(t+\delta t) \;|\; x_A(t), x_{-A}(t)\right) \\
	&= P\left(x_A(t+\delta t) \;|\;x_A(t), x_{-A}(t),  x_{-A}(t+\delta t) =  x_{-A}(t)\right) \delta^{x_{-A}(t+ \delta t)}_{x_{-A}(t)}
}
This
formalizes the statement in the text that under the rate matrix $K(A)$, $x_{-A}$ does not change
its state. 

Next, since $A$ is a community under $K(A;t)$, we can expand further to get
\eq{
P(x_A(t+\delta t), x_{-A}(t+\delta t) \;|\; x_A(t), x_{-A}(t)) 
	&= P\left(x_A(t+\delta t) \;|\;x_A(t)\right) \delta^{x_{-A}(t+ \delta t)}_{x_{-A}(t)}
}
So the full joint distribution is
\eq{
P(x_A(t+\delta t), x_{-A}(t+\delta t), x_A(t), x_{-A}(t)) 
	&= P\left(x_A(t+\delta t) \;|\;x_A(t)\right) \delta^{x_{-A}(t+ \delta t)}_{x_{-A}(t)} P(x_A(t), x_{-A}(t))
}

We can use this form of the joint distribution to establish the following two equations
\eq{
\label{eq:d3}
S_P(X_{-A}(t + \delta t) \;|\; X_{-A}(t), X_A(t+\delta t) &= 0 \\
\label{eq:d4}
S_P(X_{-A}(t) \;|\; X_{-A}(t+\delta t), X_A(t+\delta t) &= 0
}
Applying the chain rule for entropy to decompose $S_P(X_{-A}(t), X_{-A}(t+\delta t) \;|\; X_{A}(t+\delta t))$ in 
two different ways, and plugging \cref{eq:d3,eq:d4}, respectively, into those two decompositions, we see that
\eq{
S_P(X_{-A}(t+\delta t) \;|\; X_A(t+\delta t)) &= S_P(X_{-A}(t) \;|\; X_A(t+\delta t))
\label{eq:d5}
}

Next, use \cref{eq:d5} to expand
\eq{
\dfrac{d^{A;K}}{dt} S^{X | X_{A}}(t) &= \lim_{\delta t \rightarrow 0} 
	\dfrac{S_P(X_{-A}(t) \;|\; X_A(t)) - S_P(X_{-A}(t+\delta t) \;|\; X_A(t+\delta t))} {\delta t} \\
	&=  \lim_{\delta t \rightarrow 0} 
	\dfrac{S_P(X_{-A}(t) \;|\; X_A(t)) - S_P(X_{-A}(t) \;|\; X_A(t+\delta t))} {\delta t}
}
Add and subtract $S(X_{-A}(t))$ in the numerator on the RHS to get
\eq{
\dfrac{d^{A;K}}{dt} S^{X | X_{A}}(t)   &=  \lim_{\delta t \rightarrow 0} 
	\dfrac{I(X_{-A}(t) \;|\; X_A(t)) - I(X_{-A}(t) ; X_A(t+\delta t))} {\delta t}
}
Since $X_{-A}(t)$ and $X_A(t+\delta t)$ are conditionally independent given $X_A(t)$,
the difference of mutual informations in the numerator on the RHS is non-negative,
by the data-processing inequality~\cite{cover_elements_2012}. 

This completes the proof.

\section{Proof of \cref{eq:30a}}
\label{app:proof_eq_30}

For simplicity of the exposition, treat $\oo$ as though it were all of $\NN$, 
i.e., suppress the $\oo$ index in $x_\oo$ and $x'_\oo$, suppress the $\oo$ argument of $K(\oo;t)$, and implicitly restrict
sums over subsystems $i$ to elements of $\oo$. Then using the definition of $K(\oo';t)$, we can expand
\eq{
\dot{\sigma}(t) &= \sum_{x,x'} K^{x'}_{x}(t) p_{x'}(t)\ln \left[\dfrac{K^{x'}_{x}(t)}{K^{x}_{x'}(t)p_x(t)}\right] \\
	&= \sum_{x,x'} K^{x'}_{x}(\oo'; t) p_{x'}(t) \ln \left[\dfrac{K^{x'}_{x}(t)}{K^{x}_{x'}(t)p_x(t)}\right]
		+ \sum_{i \not\in \oo'}\sum_{x,x'} K^{x'}_x(i;t) p_{x'}(t) \ln \left[\dfrac{K^{x'}_{x}(t)}{K^{x}_{x'}(t)p_x(t)}\right] 
\label{eq:c2}
}
Since $\oo'$ is a community, by \cref{eq:a2} we can rewrite the first sum on the RHS of \cref{eq:c2} as
\eq{
\sum_{x_{\oo'},x'_{\oo'}} K^{x'_{\oo'}}_{x_{\oo'}}(\oo'; t) 
	p_{x'}(t) \ln \left[\dfrac{K^{x'_{\oo'}}_{x_{\oo'}}(\oo';t)}{K^{x_{\oo'}}_{x'_{\oo'}}(\oo';t)p_{x_{\oo'}}(t)}\right]
		- \sum_{x,x'} K^{x'}_x(\oo';t) p_{x'}(t) \ln \left[\dfrac{p_x(t)} {p_{x_{\oo'}}(t)  }\right]
			&= \langle {\dot{\sigma}}^{\oo'}(t) \rangle + \dfrac{d^{\oo';K}}{dt}S^{X_\oo | X_{\oo'}}(t)
}
Moreover, by \cref{eq:a1}, even though $\oo \setminus \oo'$ need not be a community,
the second sum in \cref{eq:c2} can be rewritten as
\eq{
\sum_{i \not\in \oo'} \sum_{x,x'} K^{x'}_x(i;t) p_{x'}(t) \ln \left[\dfrac{K^{x'}_x(i;t)}{K^x_{x'}(i;t)p_x(t)}\right]
	&= \sum_{x,x'} K^{x'}_x(\oo \setminus \oo';t)p_{x'}(t) 
			\ln \left[\dfrac{K^{x'}_x(\oo \setminus \oo';t)}{K^x_{x'}(\oo\setminus\oo'; t)p_x(t)}\right] \\
	&= \langle {\dot{\sigma}}_{K(\oo \setminus \oo';t)}(t) \rangle
}

Combining completes the proof. In order to express that proof as in the main text,
with the implicit $\oo$ once again made explicit, use the fact that windowing $K(\oo;t)$ to $\oo' \subset \oo$ is the same as windowing $K(t)$ to $\oo'$.

\section{EP bounds from the inclusion-exclusion principle}
\label{app:in_ex_sums}


For all $n > 1$, write $\NN^n$ for the multiset of all intersections of $n$ of the sets of subsystems
$\oo_i$:
\eq{
\NN^2 &= \{\oo_i \cap \oo_j : 1 \le i < j < |\NN^1|\} \\
\NN^3 &= \{\oo_i \cap \oo_j \cap \oo_k: 1 \le i < j < k < |\NN^1|\} 
}
and so on, up to $\NN^{|\NN^1|}$. Any community structure $\NN^1$ specifies an associated set of sets,
\eq{
\ovNN:= \NN^1 \cup \NN^2 \cup \ldots \cup \NN^{|\NN^1|}
}
Note that every element of $\ovNN$ is itself a community,
since intersections of communities are unions of communities. 

Given any function $f : \ovNN \rightarrow \R$, 
the associated \textbf{inclusion-exclusion sum} (or just ``in-ex sum'') is
\eq{
\GG^{f} &:= \sum_{\oo \in \NN^1} f(\oo) - \sum_{\oo \in \NN^2} f(\oo) + \sum_{\oo \in \NN^3}  f(\oo) - \ldots
}
In particular, given any distribution $p_x$, there is an associated real-valued function mapping
any $\oo \in \ovNN$ to the marginal entropy of (the subsystems in) $\oo$.
So using $S^{\ovNN}$ to indicate that function, 
\eq{
\GG^{S} &:=  \sum_{\oo \in \NN^1} S^{\oo} - \sum_{\oo \in \NN^2} S^{\oo} + \sum_{\oo \in \NN^3} S^{\oo}  - \ldots
\label{eq:5aa}
}
where $S^{\oo}$ is shorthand for $S^{X_\oo}$.
I refer to $\GG^{S} -S^{\NN}$ as the \textbf{in-ex information}.
As an example, if $\NN^1$ consists of two subsets, $\oo_1, \oo_2$, with no intersection, then the in-ex information is just
the mutual information $I(X_{\oo_1} ; X_{\oo_2})$. As another example, if $\NN^1$ consists of all singletons $i \in \NN$,
then the in-ex information is the multi-information of the $N$ separate random variables.

The global EP rate is the negative derivative of the in-ex information, plus the
in-ex sum of local EP rates:
\eq{
\label{eq:15a}
\dSt	&= 
\dfrac{d S^\NN(t)}{dt} + \langle\dot{Q}^\NN (t)\rangle  \\
	&= 
\dfrac{d}{dt} \left[ S^{\NN}(t) - \GG^{S(t)}\right] + \GG^{\langle\dot{\sigma}^\ovNN (t)\rangle}
\label{eq:11a}
}

\begin{proof}
To establish \cref{eq:15a}, first plug in to the result in \cref{app:proof_eq_11} and use the normalization of the rate matrices
to see that the EP rate of the full set of $N$ coupled subsystems is
\eq{
\dSt &= \sum_{x',x} K^{x'}_x(t) p_{x'}(t) \ln \left[\dfrac{K^{x'}_x(t) p_{x'}(t)}
				{K^{x}_{x'}(t) p_{x}(t)}\right]  \nonumber \\
	&= \sum_{i, x',x} K^{x'}_x(i; t) p_{x'}(t) \ln \left[\dfrac{K^{x'}_x(t) p_{x'}(t)}
				{K^{x}_{x'}(t) p_{x}(t)}\right]  \nonumber \\
	&= \sum_{i, x',x} K^{x'}_x(i; t) p_{x'}(t) \ln \left[\dfrac{K^{x'}_x(t)}
				{K^{x}_{x'}(t) p_{x}(t)}\right] 
\label{eq:A2}
}

Now introduce the shorthand
\eq{
G(A \subseteq \NN) &:= \sum_{i\in A,x',x} K^{x'}_x(i; t) p_{x'}(t) \ln \left[\dfrac{K^{x'}_x(t) }
				{K^{x}_{x'}(t) p_{x}(t)}\right]
}
Note that $\NN$ itself is a community; $G$ is an additive function over subsets of $\NN$; and 
$\dSt = G(\NN)$. Accordingly, we can
apply the inclusion-exclusion principle to \cref{eq:A2} for the set of subsets $\NN(t)$ to get
\eq{
\dSt	&= \sum_{\oo \in \NN^1(t)} G(\oo)  -  \sum_{\oo \in \NN^2(t)} G(\oo)  +  \sum_{\oo \in \NN^3(t)} G(\oo)  - \ldots \nonumber  \\
&= \sum_{\oo \in \NN^1(t)} \sum_{i  \in \oo} \sum_{x,x'} K^{x'}_x(i;t)  p_{x'}(t) \ln \left[\dfrac{K^{x'}_{x}(t)}
				{K^{x}_{x'}(t) p_{x}(t)}\right]  \nonumber \\
		& \qquad -  \sum_{\ooo \in \NN^2(t)} \sum_{i  \in \oo} \sum_{x,x'} K^{x'}_x(i;t)  p_{x'}(t)  \ln \left[\dfrac{K^{x'}_{x}(t)}
				{K^{x}_{x'}(t) p_{x}(t)}\right]   \nonumber  \\
		& \qquad +  \sum_{\ooo \in \NN^3(t)} \sum_{i  \in \oo} \sum_{x,x'} K^{x'}_x(i;t)  p_{x'}(t)  \ln \left[\dfrac{K^{x'}_{x}(t)}
				{K^{x}_{x'}(t) p_{x}(t)}\right] \nonumber \\
		&\qquad - \ldots \nonumber  \\
	&= \sum_{\oo \in \NN^1(t)}  \sum_{x,x'} K^{x'}_{x}(\oo;t)  p_{x'}(t) \ln \left[\dfrac{K^{x'}_{x}(t)}
				{K^{x}_{x'}(t) p_{x}(t)}\right]  -  \sum_{\oo \in \NN^2(t)}  \sum_{x,x'} K^{x'}_{x}(\oo;t)  p_{x'}(t) \ln \left[\dfrac{K^{x'}_{x}(t)}
				{K^{x}_{x'}(t) p_{x}(t)}\right]  + \ldots
\label{eq:a5}
}
Now use \cref{eq:a2} in \cref{lemma:a1} to rewrite \cref{eq:a5} as
\eq{
\dSt   	
	&= \sum_{\oo \in \NN^1(t)}  \sum_{x,x'} K^{x'_\oo}_{x_\oo}(\oo;t) 
	 \delta^{x'_{-\oo}}_{x_{-\oo}}p_{x'_\oo,x'_{-\oo}}(t) \ln \left[\dfrac{K^{x'_\oo}_{x_\oo}(\oo;t)}
				{K^{x_\oo}_{x'_\oo}(\oo;t)p_{x}(t)}\right]  -  \sum_{\oo \in \NN^2(t)}  \sum_{x,x'} K^{x'_\oo}_{x_\oo}(\oo;t)   \delta^{x'_{-\oo}}_{x_{-\oo}} p_{x'_\oo,x'_{-\oo}}(t) \ln \left[\dfrac{K^{x'_\oo}_{x_\oo}(\oo;t)}
				{K^{x_\oo}_{x'_\oo}(\oo;t) p_{x}(t)}\right]  + \ldots \nonumber \\
	&= \sum_{\oo \in \NN^1(t)}  \sum_{x_\oo,x_{-\oo},x'_\oo} K^{x'_\oo}_{x_\oo}(\oo;t)  p_{x'_\oo,x_{-\oo}}(t) \ln \left[\dfrac{K^{x'_\oo}_{x_\oo}(\oo;t)}
				{K^{x_\oo}_{x'_\oo}(\oo;t)p_{x_\oo,x_{-\oo}}(t)}\right] -  \sum_{\oo \in \NN^2(t)}  \sum_{x_\oo,x_{-\oo},x'_\oo} K^{x'_\oo}_{x_\oo}(\oo;t)  p_{x'_\oo,x_{-\oo}}(t) \ln \left[\dfrac{K^{x'_\oo}_{x_\oo}(\oo;t)}
				{K^{x_\oo}_{x'_\oo}(\oo;t)p_{x_\oo,x_{-\oo}}(t)}\right]  + \ldots
\label{eq:A4}
}
Next, use the same kind of reasoning that resulted in \cref{eq:A4} to show that the sum
\eq{
& \sum_{\oo \in \NN^1(t)}  \sum_{x_\oo,x_{-\oo},x'_\oo} K^{x'_\oo}_{x_\oo}(\oo;t)  p_{x'_\oo,x_{-\oo}}(t) \ln p_{x_\oo,x_{-\oo}}(t)
	 - \sum_{\oo \in \NN^2(t)}  \sum_{x_\oo,x_{-\oo},x'_\oo} 
				K^{x'_\oo}_{x_\oo}(\oo;t)  p_{x'_\oo,x_{-\oo}}(t) \ln p_{x_\oo,x_{-\oo}}(t) + \ldots 
}
can be written as $\sum_{x,x'} K^{x'}_{x}(t) p_{x'}(t) \ln p_x(t) = S^{\NN}(t)$.
We can use this to rewrite \cref{eq:A4} as
\eq{
\dSt	&=\dfrac{d}{dt} \left[ S^{X_\NN}(t) - \GG^{S(t)}\right] 
	 + \sum_{\oo \in \ovNN}  \sum_{x_\oo,x'_\oo} K^{x'_\oo}_{x_\oo}(\oo;t)  p_{x'_\oo}(t) \ln \left[\dfrac{K^{x'_\oo}_{x_\oo}(t)}
				{K^{x_\oo}_{x'_\oo}(t) p_{x_\oo}(t)}\right] 
	-  \sum_{\oo \in \NN^2(t)}  \sum_{x_\oo,x'_\oo} K^{x'_\oo}_{x_\oo}(\oo;t)  p_{x'_\oo}(t) \ln \left[\dfrac{K^{x'_\oo}_{x_\oo}(\oo;t)}
				{K^{x_\oo}_{x'_\oo}(t) p_{x_\oo}(t)}\right]  + \ldots \nonumber \\
	&=  \dfrac{d}{dt} \left[ S^{X_\NN}(t) - \GG^{S(t)}\right]  + \sum_{\oo \in \ovNN} \dSot -  \sum_{\oo \in \NN^2(t)} \dSot + \ldots
}

This establishes the claim. 

\end{proof}

If we use \cref{eq:15} to expand each local EP term in \cref{eq:11a} and then compare
to \cref{eq:15a}, we see that the global expected EF rate equals the in-ex sum of the local expected EF rates:
\eq{
\langle\dot{Q}^\NN (t)\rangle &= \GG^{\langle\dot{Q}^\ovNN (t)\rangle}
}
Note as well as that we can apply \cref{eq:11a} to itself, by using it to expand any of the local EP terms
$\sigma^\oo(t)$ that occur in the in-ex sum $\GG^{\langle\dot{\sigma}^\ovNN (t)\rangle}$ on its own RHS.

\cref{eq:11a} 
can be particularly useful when combined
with the fact that for any two communities $\oo, \oo' \subset \oo$, 
$\langle \dot{\sigma}^{\oo'}  \rangle \le \langle \dot{\sigma}^{\oo}  \rangle$ (see \cref{eq:30a}). To illustrate this, return to the scenario of \cref{fig:1}.
There are three communities in $\NN^1$ (namely, $\{1, 2, 3\}, \{3\}, \{3, 4\}$), three 
in $\NN^2$ (namely, three copies of $\{3\}$), and one in $\NN^3$ (namely, $\{3\}$). Therefore using obvious shorthand,
\eq{
\dSt	&= \dfrac{d S^{1,2,3,4}(t)}{dt} - \dfrac{d S^{1,2,3}(t)}{dt} - \dfrac{d S^{3,4}(t)}{dt} + \dfrac{d S^{3}(t)}{dt}  \nonumber \\
	&\qquad\qquad\qquad\qquad\qquad\qquad
			 + \langle \dot{\sigma}^{1,2,3}  \rangle + \langle \dot{\sigma}^{3,4}\rangle  - \langle \dot{\sigma}^{3}\rangle   \nonumber \\
	&\ge \dfrac{d S^{4 | 1,2,3}(t)}{dt} - \dfrac{d S^{4 | 3}(t)}{dt} + \langle \dot{\sigma}^{1,2,3}  \rangle
\label{eq:39}
}
(Note that in contrast to lower bounds involving windowed derivatives, none of the terms in \cref{eq:39} involve
counterfactual rate matrices.)
So if the entropy of subsystem $4$ conditioned on subsystems $1, 2$ and $3$ is growing, while
its entropy conditioned on only subsystem $3$ is shrinking, then the global EP rate is strictly positive.

As a final comment, it is worth noting that in contrast to multi-information, in some situations the in-ex information can be negative.
(In this it is just like some other  extensions of mutual information to more
than two variables~\cite{mcgill1954multivariate,ting1962amount}.)
As an example, suppose $N=6$, and label the subsystems as $\NN = \{12, 13, 14, 23, 24, 34\}$.
Then take $\NN^1$ to have four elements, $\{12, 13, 14\}, \{23, 24, 12\}, \{34, 13, 23\}$ and
$\{34, 24, 14\}$. (So the first element consists of all subsystems whose label involves a $1$,
the second consists of all subsystems whose label involves a $2$, etc.). 
Also suppose that with probability $1$, the state of every subsystem is the same. Then if the probability distribution of that identical 
state is $p$, the in-ex information is $-S(p) + 4S(p) - 6S(p) = -3S(p) \le 0$.

\end{widetext}

\end{document}